\documentclass[11pt]{article}
\usepackage{amssymb,mathrsfs}
\usepackage{amsmath,amsfonts,amssymb,amsthm}
\usepackage{xspace}
\usepackage{graphicx}
\usepackage{fullpage}
\usepackage{latexsym}
\usepackage{algorithm}
\usepackage{algorithmic}
\usepackage{epstopdf}

\newtheorem{theorem}{Theorem}
\newtheorem{lemma}{Lemma}

\newtheorem{corollary}{Corollary}

\newtheorem{problem}{Problem}

\newtheorem{fact}{Fact}

\newcommand{\eat}[1]{}

\newcommand{\calP}{{\mathcal P}}

\newcommand{\calT}{{\mathcal T}}

\newcommand{\size}[1]{|#1|\xspace}

\def \qas{quasi-identifiers\xspace}
\def \sa{sensitive attribute\xspace}
\def \sas{sensitive attributes\xspace}
\def \saspace{SA space\xspace}
\def \kanony{\textsc{$k$-Anonymity}\xspace}
\def \tclose{\textsc{$t$-Closeness}\xspace}
\def \vectfull{sensitive attribute value distribution\xspace}
\def \vect{SA distribution\xspace}
\def \vects{SA distributions\xspace}
\def \prob{\mathbf{P}\xspace}
\def \emd{{\rm EMD}\xspace}
\def \group{group\xspace}
\def \groups{groups\xspace}
\def \calT{\mathcal{T}}
\def \calP{\mathcal{P}}

\begin{document}

\title{On the Complexity of $t$-Closeness Anonymization and Related Problems}
\author{Hongyu Liang\footnotemark[1] \and Hao Yuan\footnotemark[2]}

\renewcommand{\thefootnote}{\fnsymbol{footnote}}

\footnotetext[1]{Institute for Interdisciplinary Information
Sciences, Tsinghua University, Beijing, China.
Email: \texttt{lianghy08@mails.tsinghua.edu.cn}. Supported in part by the National
Basic Research Program of China Grant 2011CBA00300, 2011CBA00301,
and the National Natural Science Foundation of China Grant 61033001,
61061130540, 61073174.
}

\footnotetext[2]{Department of Computer Science, City University of Hong Kong, Kowloon, Hong Kong, China.
Email: \texttt{haoyuan@cityu.edu.hk}. Supported by the Research Grants Council of Hong Kong under grant 9041688 (CityU 124411).}

\date{}
\maketitle

\begin{abstract}
An important issue in releasing individual data is to protect the sensitive information from being leaked and maliciously utilized. Famous privacy preserving principles that aim to ensure both data privacy and data integrity, such as $k$-anonymity and $l$-diversity, have been extensively studied both theoretically and empirically. Nonetheless, these widely-adopted principles are still insufficient to prevent attribute disclosure if the attacker has partial knowledge about the overall sensitive data distribution. The $t$-closeness principle has been proposed to fix this, which also has the benefit of supporting numerical sensitive attributes. However, in contrast to $k$-anonymity and $l$-diversity, the theoretical aspect of $t$-closeness has not been well investigated.

We initiate the first systematic theoretical study on the $t$-closeness principle under the commonly-used attribute
suppression model. We prove that for every constant $t$ such that $0\leq t<1$, it is NP-hard to find an optimal
$t$-closeness generalization of a given table. The proof consists of several reductions each of which works for
different values of $t$, which together cover the full range. To complement this negative result,
we also provide exact and fixed-parameter algorithms. Finally, we answer some open questions regarding the
complexity of $k$-anonymity and $l$-diversity left in the literature.
\end{abstract}

\section{Introduction}
Privacy-preserving data publication is an important and active topic in the database area. Nowadays many organizations need to publish microdata that contain certain information, e.g., medical condition, salary, or census data, of a collection of individuals, which are very useful for research and other purposes. Such microdata are usually released as a table, in which each record (i.e., row) corresponds to a particular individual and each column represents an attribute of the individuals. The released data usually contain \emph{sensitive attributes}, such as Disease and Salary, which, once leaked to unauthorized parties, could be maliciously utilized and harm the individuals.
Therefore, those features that can directly identify individuals, e.g., Name and Social Security Number, should be removed from the released table. See Table~\ref{tab:1} for example of an (imagined) microdata table that a hospital prepares to release for medical research. (Note that the IDs in the first column are only for simplicity of reference, but not part of the table.)

\begin{table}[t]
\centering
\begin{tabular}{|c||c|c|c||c|}
\hline
& \multicolumn{3}{c||}{Quasi-identifiers} & Sensitive\\
\hline
& Zipcode & Age& Education&Disease\\
\hline
1 & 98765 & 38 & Bachelor & Viral Infection\\
2 & 98654 & 39 & Doctorate & Heart Disease\\
3 & 98543 & 32 & Master & Heart Disease\\
4 & 97654 & 65 & Bachelor & Cancer\\
5 & 96689 & 45 & Bachelor & Viral Infection\\
6 & 97427 & 33 & Bachelor & Viral Infection\\
7 & 96552 & 54 & Bachelor & Heart Disease\\
8 & 97017 & 69 & Doctorate & Cancer\\
9 & 97023 & 55 & Master & Cancer\\
10 & 97009 & 62 & Bachelor & Cancer\\
\hline
\end{tabular}
\caption{The raw microdata table.}\label{tab:1}
\end{table}

\begin{table}
\centering
\begin{tabular}{|c||c|c|c||c|}
\hline
& \multicolumn{3}{c||}{Quasi-identifiers} & Sensitive\\
\hline
& Zipcode & Age& Education&Disease\\
\hline
1 & 98$\star$$\star$$\star$ & 3$\star$ & $\star$ & Viral Infection\\
2 & 98$\star$$\star$$\star$ & 3$\star$ & $\star$ & Heart Disease\\
3 & 98$\star$$\star$$\star$ & 3$\star$ & $\star$ & Heart Disease\\
\hline
4 & 9$\star$$\star$$\star$$\star$ & $\star$$\star$ & Bachelor & Cancer\\
5 & 9$\star$$\star$$\star$$\star$ & $\star$$\star$ & Bachelor & Viral Infection\\
6 & 9$\star$$\star$$\star$$\star$ & $\star$$\star$ & Bachelor & Viral Infection\\
7 & 9$\star$$\star$$\star$$\star$ & $\star$$\star$ & Bachelor & Heart Disease\\
\hline
8 & 970$\star$$\star$ & $\star$$\star$ & $\star$ & Cancer\\
9 & 970$\star$$\star$ & $\star$$\star$ & $\star$ & Cancer\\
10 & 970$\star$$\star$ & $\star$$\star$ & $\star$ & Cancer\\
\hline
\end{tabular}
\caption{A 3-anonymous partition.}\label{tab:2}
\end{table}

\begin{table}
\centering
\begin{tabular}{|c||c|c|c||c|}
\hline
& \multicolumn{3}{c||}{Quasi-identifiers} & Sensitive\\
\hline
& Zip Code & Age& Education&Disease\\
\hline
1 & 98$\star$$\star$$\star$ & 3$\star$ & $\star$ & Viral Infection\\
2 & 98$\star$$\star$$\star$ & 3$\star$ & $\star$ & Heart Disease\\
\hline
3 & 9$\star$$\star$$\star$$\star$ & $\star$$\star$ & $\star$ & Heart Disease\\
5 & 9$\star$$\star$$\star$$\star$ & $\star$$\star$ & $\star$ & Viral Infection\\
8 & 9$\star$$\star$$\star$$\star$ & $\star$$\star$ & $\star$ & Cancer\\
9 & 9$\star$$\star$$\star$$\star$ & $\star$$\star$ & $\star$ & Cancer\\
\hline
4 & 97$\star$$\star$$\star$ & $\star$$\star$ & Bachelor & Cancer\\
6 & 97$\star$$\star$$\star$ & $\star$$\star$ & Bachelor & Viral Infection\\
\hline
7 & 9$\star$$\star$$\star$$\star$ & $\star$$\star$ & Bachelor & Heart Disease\\
10 & 9$\star$$\star$$\star$$\star$ & $\star$$\star$ & Bachelor & Cancer\\
\hline
\end{tabular}
\caption{A 2-diverse partition of Table 1.}\label{tab:3}
\end{table}

Nonetheless, even with unique identifiers removed from the table, sensitive personal information can still be disclosed due to the \emph{linking attacks} \cite{kanonymity2,kanonymity}, which try to identify individuals from the combination of \emph{quasi-identifiers}. The quasi-identifiers are those attributes that can reveal partial information of the individual, such as Gender, Age, and Hometown. For instance, consider an adversary who knows that one of the records in Table~\ref{tab:1} corresponds to Bob. In addition he knows that Bob is around thirty years old and has a Master's Degree. Then he can easily identify the third record as Bob's and thus learns that Bob has a heart disease.

A widely-adopted approach for protecting privacy against such attacks is \emph{generalization}, which partitions the records into disjoint groups and then transforms the quasi-identifier values in each group to the same form. (The sensitive attribute values are not generalized because they are usually the most important data for research.) Such generalization needs to satisfy some \emph{anonymization principles}, which are designed to guarantee data privacy to a certain extent.

The earliest (and probably most famous) anonymization principle is the \emph{$k$-anonymity} principle proposed by Samarati \cite{kanonymity2} and Sweeney \cite{kanonymity}, which requires each group in the partition to have size at least $k$ for some pre-specified value of $k$; such a partition is called \emph{$k$-anonymous}. Intuitively, this principle ensures that every combination of quasi-identifier values appeared in the table is indistinguishable from at least $k-1$ other records, and hence protects the individuals from being uniquely recognized by linking attacks.
The $k$-anonymity principle has been extensively studied, partly due to the simplicity of its statement. Table~\ref{tab:2} is an example of a 3-anonymous partition of Table~\ref{tab:1}, which applies the commonly-used \emph{suppression method} to generalize the values in the same group, i.e., \emph{suppresses} the conflicting values with a new symbol `$\star$'.

A potential issue with the $k$-anonymity principle is that it is totally independent of the sensitive attribute values. This issue was formally raised by Machanavajjhala et al. \cite{ldiversity} who showed that $k$-anonymity is insufficient to prevent disclosure of sensitive values against the \emph{homogeneity attack}. For example, assume that an attacker knows that one record of Table~\ref{tab:2} corresponds to Danny, who is an elder with a Doctorate Degree. From Table~\ref{tab:2} he can easily conclude that Danny's record must belong to the third group, and hence knows Danny has a cancer since all people in the third group have the same disease. To forestall such attacks, Machanavajjhala et al. \cite{ldiversity} proposed the \emph{$l$-diversity} principle, which demands that at most a $1/l$ fraction of the records can have the same sensitive value in each group; such a partition is called \emph{$l$-diverse}. Table~\ref{tab:3} is an example of a 2-diverse partition of Table~\ref{tab:1}. (There are some other formulations of $l$-diversity, e.g., one requiring that each group comprises at least $l$ different sensitive values.)

Li et al. \cite{icde07_tclose} observed that the $l$-diversity principle is still insufficient to protect sensitive information disclosure against the \emph{skewness attack}, in which the attacker has partial knowledge of the \emph{overall} sensitive value distribution. Moreover, since $l$-diversity only cares whether two sensitive values are distinct or not, it fails to well support sensitive attributes with semantic similarities, such as numerical attributes (e.g., the salary).

To fix these drawbacks, Li et al. \cite{icde07_tclose} introduced the \emph{$t$-closeness} principle, which requires that the sensitive value distribution in any group differs from the overall sensitive value distribution by at most a threshold $t$. There is a metric space defined on the set of possible sensitive values, in which the maximum distance of two points (i.e., sensitive values) in the space is normalized to 1.
The distance between two probability distributions of sensitive values are then measured by the \emph{Earth-Mover Distance} (EMD) \cite{emd}, which is widely used in many areas of computer science.
Intuitively,
the EMD measures the minimum amount of work needed to transform one probability distribution to another by means of moving
distribution mass between points in the probability space. The EMD between two distributions in the (normalized) space is always between 0 and 1. We will give an example of a $t$-closeness partition of Table~\ref{tab:1} for some threshold $t$ later in Section~\ref{sec:pre}, after the related notation and definitions are formally introduced.

The $t$-closeness principle has been widely acknowledged as an enhanced principle that fixes the main drawbacks of previous approaches like $k$-anonymity and $l$-diversity. There are also many other principles proposed to deal with different attacks or for use of ad-hoc applications; see, e.g., \cite{cikm11_BLLW,icde07_worstcase,kdd11,sigmod07_NAC,sigmod06_XT,sigmod07_XT,cikm11_XKRP,icde07_ZKSY} and the references therein.

\subsection{Theoretical Models of Anonymization}
It is always assumed that the released table itself satisfies the considered principle ($k$-anonymity, $l$-diversity, or $t$-closeness), since otherwise there exists no feasible solution at all. Therefore, the trivial partition that puts all records in a single group always guarantees the principle to be met. However, such a solution is useless in real-world scenarios, since it will most probably produce a table full of `$\star$'s, which is undesirable in most applications. This extreme example demonstrates the importance of finding a balance between data privacy and \emph{data integrity}.

Meyerson and Williams \cite{pods04_kanony} proposed a framework for theoretically measuring the data integrity, which aims to find a partition (under certain constraints such as $k$-anonymity) that minimizes the number of suppressed cells (i.e., `$\star$'s) in the table. This model has been widely adopted for theoretical investigations of anonymization principles. Under this model, $k$-anonymity and $l$-diversity have been extensively studied; more detailed literature reviews will be given later.

However, in contrast to $k$-anonymity and $l$-diversity, the theoretical aspects of the $t$-closeness principle have not been well explored before. There are only a handful of algorithms designed for achieving $t$-closeness \cite{icde07_tclose,tkde10_tclose,sabre,alg_tclose_inf_theory}.
The algorithms given by Li et al. \cite{icde07_tclose,tkde10_tclose} incorporate $t$-closeness into  $k$-anonymization frameworks (Incognito \cite{incognito} and Mondrian \cite{mondrian}) to find a $t$-closeness partition. Cao et al. \cite{sabre} proposed the SABRE algorithm, which is the first framework tailored for $t$-closeness. The information-theoretic approach in \cite{alg_tclose_inf_theory} works for an ``average'' version of $t$-closeness.
None of these algorithms is guaranteed to have good worst-case performance. Furthermore, to the best of our knowledge, no computational complexity results of $t$-closeness have been reported in the literature.

\subsection{Our Contributions}
In this paper, we initiate the first systematic theoretical study on the $t$-closeness principle under the commonly-used suppression framework. First, we prove that for every constant $t$ such that $0\leq t<1$, it is NP-hard to find an optimal
$t$-closeness generalization of a given table. Notice that the problem becomes trivial when $t=1$, since the EMD between any two sensitive value distributions is at most 1, and hence putting each record in a distinct group provides a feasible solution that does not need to suppress any value at all, which is of course optimal. Our result shows that the problem immediately becomes hard even if the threshold is relaxed to, say, 0.999. At the other extreme, a 0-closeness partition demands that the sensitive value distribution in every group must be the same with the overall distribution. This seems to restrict the sets of feasible solutions in a very strong sense, and thus one might imagine whether there exists an efficient algorithm for dealing with this special case. Our result dashes the hope for this idea.
The proof of our hardness result actually consists of several different reductions. Interestingly, each of these reductions only work for a set of special values of $t$, but altogether they cover the full range $[0,1)$. We note that the hardness of $t_1$-closeness does not directly imply that of $t_2$-closeness for $t_1\neq t_2$, since they may have very different optimal objective values.

As a by-product of our proof, we establish the NP-hardness of $k$-anonymity when $k=cn$, where $n$ is the number of records and $c$ is any constant in $(0,1/2]$. To the best of our knowledge, this is the first hardness result for $k$-anonymity that works for $k=\Omega(n)$. The existing approaches for proving hardness of $k$-anonymity all fail to generalize to this range of $k$ due to inherent limits of the reductions. We note that $k=n/2$ is the largest possible value for which $k$-anonymity can be hard, because when $k>n/2$, any $k$-anonymous partition can only contain one group, namely the table itself.

To complement our negative results, we also provide exact and fixed-parameter algorithms for obtaining the optimal $t$-closeness partition. Our exact algorithm for $t$-closeness runs in time $2^{O(n)}\cdot O(m)$, where $n$ and $m$ are respectively the number of rows and columns in the input table. Together with a reduction that we derive (Lemma~\ref{lem:reduction}), this gives a $2^{O(n)}\cdot O(m)$ time algorithm for $k$-anonymity for \emph{all} values of $k$, thus generalizing the result in \cite{icalp10_kanony} which only works for constant $k$. We then prove that the problem is fixed-parameter tractable when parameterized by $m$ and the alphabet size of the input table. This implies that an optimal $t$-closeness partition can be found in polynomial time if the number of quasi-identifiers and that of distinct attribute values are both small (say, constants), which is true in many real-world applications.
(We say a problem is \emph{fixed-parameter tractable} with respect to some parameters $k_1,\ldots,k_r$, if there is an algorithm solving the problem that runs in time $h(k_1,\ldots,k_r)n^{O(1)}$, where $n$ is the size of the input and $h$ is an arbitrary computable function depending only on the parameters. Parameterized complexity has become a very active research area. For standard notation and definitions in parameterized complexity, we refer the reader to \cite{parabook}.) We obtain our fixed-parameter algorithm by reducing $t$-closeness to a special \emph{mixed integer linear program} in which some variables are required to take integer values while others are not. The integer linear program we derived for characterizing $t$-closeness may have its own interest in future applications.
We note that both of our algorithms work for all values of $t$.

Last but not least, we review the problems of finding optimal $k$-anonymous and $l$-diverse partitions, and answer two open questions left in the literature.
\begin{itemize}
\item
We prove that the 2-diversity problem can be solved in polynomial time, which complements the NP-hardness results for $l\geq 3$ given in \cite{edbt10_ldiversity}. (We notice that the authors of \cite{tcs12_ldiversity} claimed that 2-diversity was proved to be polynomial by \cite{edbt10_ldiversity}. However what \cite{edbt10_ldiversity} actually proved is that the special 2-diversity instances, in which there are only two distinct sensitive values, can be reduced to the matching problem and hence solved in polynomial time. They do not have results for general 2-diversity. To the best of our knowledge, ours is the first work to demonstrate the tractability of 2-diversity.)
\item
We then present an $m$-approximation algorithm for $k$-anonymity that runs in polynomial time for all values of $k$. (Recall that $m$ is the number of quasi-identifiers.) This improves the $O(k)$ and $O(\log k)$ ratios in \cite{icdt05_kanony,sigmod07_apx_kanony} when $k$ is relatively large compared to $m$. We note that the performance guarantee of their algorithms cannot be reduced even for small values of $m$, due to some intrinsic limitations (for example, \cite{sigmod07_apx_kanony} uses the tight $\Theta(\log k)$ approximation for $k$-set cover).
\end{itemize}

\subsection{Related Work}
It is known that finding an optimal $k$-anonymous partition of a given table is NP-hard for every fixed integer $k\geq 3$ \cite{pods04_kanony}, while it can be solved optimally in polynomial time when $k \leq 2$ \cite{icalp10_kanony}. The NP-hardness result holds even for very restricted cases, e.g., when $k=3$ and there are only three quasi-identifiers \cite{joco11_kanony,joco_para_kanony}. On the other hand, Blocki and Williams \cite{icalp10_kanony} gave a $2^{O(n)}\cdot O(m)$ time algorithm that finds an optimal $k$-anonymous partition when $k=O(1)$, where $n$ and $m$ are the number of records and attributes (i.e., rows and columns) of the input table respectively. They also showed this problem to be fixed-parameter tractable when $m$ and $|\Sigma|$ (the alphabet size of the table) are considered as parameters. The parameterized complexity of $k$-anonymity has also been studied in \cite{joco_para_kanony,fct11,joco09} with respect to different parameters.

Meyerson and Williams \cite{pods04_kanony} gave an $O(k\log k)$ approximation algorithm for $k$-anonymity, i.e., it finds a $k$-anonymous partition in which the number of suppressed cells is at most $O(k\log k)$ times the optimum.
The ratio was later improved to $O(k)$ by Aggarwal et al. \cite{icdt05_kanony} and to $O(\log k)$ by Park and Shim \cite{sigmod07_apx_kanony}. We note that the algorithms in \cite{pods04_kanony,sigmod07_apx_kanony} run in time $n^{O(k)}$, and hence are guaranteed to be polynomial only if $k=O(1)$, while the algorithm in \cite{icdt05_kanony} has a truly polynomial running time for all $k$.
There are also a number of heuristic algorithms for $k$-anonymity (e.g., Incognito \cite{incognito}), which work well in many real datasets but have poor worst-case performance guarantee.

Xiao et al. \cite{edbt10_ldiversity} are the first to establish a systematic theoretical study on $l$-diversity. They showed that finding an optimal $l$-diverse partition is NP-hard for every fixed integer $l\geq 3$ even if $m$, the number of quasi-identifiers, is any fixed integer not smaller than $l$. They also provided an $(l\cdot m)$-approximation algorithm. Dondi et al. \cite{tcs12_ldiversity} proved an inapproximability factor of $c\ln (l)$ for $l$-diversity where $c>0$ is some constant, and showed that the problem remains APX-hard even if $l=4$ and the table consists of only three columns. They also presented an $m$-approximation algorithm when the number of distinct sensitive values is constant, and gave some parameterized hardness results and algorithms.

\subsection{Paper Organization}
The rest of this paper is organized as follows. Section~\ref{sec:pre} introduces notation and definitions used throughout the paper, and then formally defines the problems. Section~\ref{sec:nphard} is devoted to proving the hardness of finding the optimal $t$-closeness partition, while Section~\ref{sec:exact} provides exact and parameterized algorithms. Sections~\ref{sec:kanony} and \ref{sec:2diversity} present our results for $k$-anonymity and 2-diversity, respectively. Finally, the paper is concluded in Section~\ref{sec:conclu} with some discussions and future research directions.

\section{Preliminaries}\label{sec:pre}
We consider a raw database that contains $m$ \qas (QIs) and a \sa (SA).\footnote{Following previous approaches, we only consider instances with one sensitive attribute. Our hardness result indicates that one sensitive attribute already makes the problem NP-hard. Meanwhile, it is easy to verify that our algorithms also work for the case where multiple sensitive attributes exist.}
Each record $t$ in the database is an $(m+1)$-dimensional vector drawn from $\Sigma^{m+1}$,
where $\Sigma$ is the alphabet of possible values of the attributes.
For $1\leq i\leq m$, $t[i]$ is the value of the $i$-th QI of $t$, and $t[m+1]$ is the value of the SA of $t$.
Let $\Sigma_s \subseteq \Sigma$ be the alphabet of possible SA values. A microdata table (or table, for short) $\mathcal{T}$
is a multiset of vectors (or rows) chosen from $\Sigma^{m+1}$, and we denote by $\size{\mathcal{T}}$ the size of $\mathcal{T}$, i.e., the number of vectors contained in $\mathcal{T}$. We will let $n=\size{\mathcal{T}}$ when the table $\mathcal{T}$ is clear in the context. Note that $\mathcal{T}$ may contain identical vectors since it can be a multiset. We also use $\mathcal{T}[j]$ to denote the $j$-th vector in $\mathcal{T}$ under some ordering, e.g., $\mathcal{T}[3][m+1]$ is the SA value of the third vector of $\mathcal{T}$.

\begin{table}[t]
\centering
\begin{tabular}{|c||c|c|c||c|}
\hline
& \multicolumn{3}{c||}{Quasi-identifiers} & Sensitive\\
\hline
& Zip Code & Age& Education&Disease\\
\hline
1 & 98765 & 38 & Bachelor & Viral Infection\\
2 & 98654 & 39 & Doctorate & Heart Disease\\
3 & 98543 & 32 & Master & Heart Disease\\
\hline
\multicolumn{5}{c}{\textbf{After generalization:}}\\
\hline
1 & 98$\star$$\star$$\star$ & 3$\star$ & $\star$ & Viral Infection\\
2 & 98$\star$$\star$$\star$ & 3$\star$ & $\star$ & Heart Disease\\
3 & 98$\star$$\star$$\star$ & 3$\star$ & $\star$ & Heart Disease\\
\hline
\end{tabular}
\caption{The first three records in Table 1.}\label{tab:sub}
\end{table}

Let $\star$ be a fresh character not in $\Sigma$.
For each vector $t\in \mathcal{T}$, let $t^*$ be the \emph{suppressor} of $t$ (inside $\calT$) defined as follows:
\begin{itemize}
\item $t^*[m+1]=t[m+1]$;
\item for $1\leq i\leq m$, $t^*[i]=t[i]$ if $t[i]=t'[i]$ for all
$t'\in \mathcal{T}$, and $t^*[i]=\star$ otherwise.
\end{itemize}

The \emph{cost} of a suppressor $t^*$ is $cost(t^*)=|\{1\leq i\leq m~|~t^*[i]=\star\}|$, i.e., the number of `$\star$'s in $t^*$. It is easy to see that all vectors in $\calT$ have the same suppressor if we only consider the quasi-identifiers.
The \emph{generalization} of $\mathcal{T}$ is defined as
$Gen(\mathcal{T})=\{t^*~|~t\in \mathcal{T}\}$. (Note that $Gen(\mathcal{T})$ is also a multiset.)
The \emph{cost} of the generalization of $\calT$ is $cost(\mathcal{T})=\sum_{t^*\in Gen(\mathcal{T})}cost(t^*)$,
i.e., the sum of costs of all the suppressors. Since all suppressors in $\calT$ have the same cost, we can equivalently write $cost(\mathcal{T})=\size{\mathcal{T}}\cdot cost(t^*)$ for any
$t^* \in Gen(\mathcal{T})$.

As an illustrative example, Table~\ref{tab:sub} consists of the first three record of Table~\ref{tab:1}, which contains eight QIs (we regard each digit of Zip-code and Age as a separate QI) and one SA. The generalization of Table~\ref{tab:sub} is also shown. In this case all suppressors have cost 5, and the cost of this generalization is $5\cdot 3=15$.

A \emph{partition} $\mathcal{P}$ of table $\mathcal{T}$ is a collection of pairwise disjoint non-empty subsets of $\mathcal{T}$ whose union equals $\mathcal{T}$. Each subset in the partition is called a \emph{group} or a \emph{sub-table}. The cost of the partition $\mathcal{P}$, denoted by $cost(\mathcal{P})$, is the sum of costs of all its \groups.
For example, the partition of Table~\ref{tab:1} given by Table~\ref{tab:2} has cost $5\cdot 3+6\cdot 4+5\cdot 3=54$.

\subsection{$t$-Closeness Principle}

We formally define the $t$-closeness principle introduced in \cite{icde07_tclose} for protecting data privacy. Let $\calT$ be a table, and assume without loss of generality that $\Sigma_s=\{1,2,\ldots,|\Sigma_s|\}$.
The \emph{sensitive attribute value space} (\saspace) is a normalized metric space $(\Sigma_s,d)$, where $d(\cdot,\cdot)$ is a distance
function defined on $\Sigma_s \times \Sigma_s$ satisfying that (1)$d(i,i)=0$ for any $i\in\Sigma_s$; (2)$d(i,j)=d(j,i)$ for all $i,j\in\Sigma_s$;
(3)$d(i,j)+d(j,k)\geq d(i,k)$ for $i,j,k\in \Sigma_s$ (this is called the triangle inequality); and (4)$\max_{i,j\in \Sigma_s}d(i,j)=1$ (this is called the normalized condition).

For a sub-table $M \subseteq \calT$ and $i\in \Sigma_s$, denote by $n(M,i)$ the number of vectors whose SA value equals $i$.
Clearly $|M|=\sum_{i\in\Sigma_s}n(M,i)$.
The \emph{\vectfull} (\vect) of $M$, denoted by $\prob(M)$, is a $|\Sigma_s|$-dimensional vector whose $i$-th coordinate is
$\prob(M)[i]=n(M,i)/|M|$ for $1\leq i\leq |\Sigma_s|$.
Thus $\prob(M)$ can be seen as the probability distribution of the SA values in $M$, assuming that each vector in $M$ appears with the same probability.
For a threshold $0\leq t\leq 1$, we say $M$ have \emph{$t$-closeness} (with $\calT$) if $\emd(\prob(M),\prob(\calT))\leq t$,
where $\emd(\mathbf{X},\mathbf{Y})$ is the \emph{Earth-Mover Distance} (EMD) between distributions $\mathbf{X}$ and $\mathbf{Y}$ \cite{emd}. A \emph{$t$-closeness partition} of $\calT$ is one in which
every \group has $t$-closeness with $\calT$.

Intuitively,
the EMD measures the minimum amount of work needed to transform one probability distribution to another by means of moving
distribution mass between points in the probability space; here a unit of work corresponds to moving a unit amount of probability mass by a unit of ground distance.
The EMD between two \vects $\mathbf{X}$ and $\mathbf{Y}$ can be formally defined as the optimal objective value of the following linear program \cite{emd,icde07_tclose}:
\begin{eqnarray*}
\textrm{Minimize~}\sum_{i=1}^{|\Sigma_s|}\sum_{j=1}^{|\Sigma_s|}d(i,j)f(i,j)& &\textrm{subject to:}\\
\sum_{j=1}^{|\Sigma_s|}f(i,j)=\mathbf{X}[i],& &\forall 1\leq i\leq |\Sigma_s|\\
\sum_{i=1}^{|\Sigma_s|}f(i,j)=\mathbf{Y}[j],& &\forall 1\leq j\leq |\Sigma_s|\\
f(i,j)\geq 0,& &\forall 1\leq i,j\leq |\Sigma_s|.
\end{eqnarray*}

The above constraints are a little different from those in \cite{icde07_tclose}; however they can be proved equivalent using the triangle inequality condition of the \saspace.
It is also easy to see that $\emd(\mathbf{X},\mathbf{Y})=\emd(\mathbf{Y},\mathbf{X})$.
By the normalized condition of the \saspace, we have $0\leq \emd(\mathbf{X},\mathbf{Y})\leq 1$ for any \vects $\mathbf{X}$ and $\mathbf{Y}$.

The \emph{equal-distance space} refers to a special \saspace in which each pair of distinct sensitive values have distance exactly 1. There is a concise formula for computing the EMD between two \vects in this space.
\begin{fact}[\cite{icde07_tclose}]
For any two \vects $\mathbf{X}$ and $\mathbf{Y}$ in the equal-distance space, we have
\begin{eqnarray*}
\emd(\mathbf{X},\mathbf{Y})=\frac{1}{2}\sum_{i=1}^{|\Sigma_s|}\left|\mathbf{X}[i]-\mathbf{Y}[i]\right|
=\sum_{1\leq i\leq |\Sigma_s|:\mathbf{X}[i]\geq \mathbf{Y}[i]}(\mathbf{X}[i]-\mathbf{Y}[i]).
\end{eqnarray*}
\end{fact}
Therefore, in the equal-distance space, the EMD coincides with the \emph{total variation distance} between two distributions.

Let us go back to Table~\ref{tab:1} for an example. We let 1,2,and 3 denote the sensitive values ``Viral Inspection'', ``Heart Disease'', and ``Cancer'', respectively. Let the \saspace be the equal-distance space. The \vect of the whole table is then $(0.3,0.3,0.4)$.
Suppose we set the threshold $t=0.3$. It can be verified that Table~\ref{tab:3}, although being a 2-diverse partition, is not a $t$-closeness partition of Table~\ref{tab:1}. In fact, the \vect of the first group is $(0.5,0.5,0)$, and hence the EMD between it and the overall distribution is 0.4. (This example also reflects some property of the skewness attack that $l$-diversity suffers from. If an attacker can locate the record of Alice in the first group of Table~\ref{tab:3}, then he knows that Alice does not have a cancer. If he in addition knows that Alice comes from some district where people have a very low chance to have heart disease, then he would be confident that Alice has a viral infection.)
We instead give a $0.3$-closeness partition in Table~\ref{tab:tclose}. We can actually verify that it is even a 0.1-closeness partition.

\begin{table}[t]
\centering
\begin{tabular}{|c||c|c|c||c|}
\hline
& \multicolumn{3}{c||}{Quasi-identifiers} & Sensitive\\
\hline
& Zipcode & Age& Education&Disease\\
\hline
1 & 9$\star$$\star$$\star$$\star$ & $\star$$\star$ & $\star$ & Viral Infection\\
2 & 9$\star$$\star$$\star$$\star$ & $\star$$\star$ & $\star$ & Heart Disease\\
4 & 9$\star$$\star$$\star$$\star$ & $\star$$\star$ & $\star$ & Cancer\\
\hline
3 & 9$\star$$\star$$\star$$\star$ & $\star$$\star$ & $\star$ & Heart Disease\\
5 & 9$\star$$\star$$\star$$\star$ & $\star$$\star$ & $\star$ & Viral Infection\\
8 & 9$\star$$\star$$\star$$\star$ & $\star$$\star$ & $\star$ & Cancer\\
9 & 9$\star$$\star$$\star$$\star$ & $\star$$\star$ & $\star$ & Cancer\\
\hline
6 & 9$\star$$\star$$\star$$\star$ & $\star$$\star$ & Bachelor & Viral Infection\\
7 & 9$\star$$\star$$\star$$\star$ & $\star$$\star$ & Bachelor & Heart Disease\\
10 & 9$\star$$\star$$\star$$\star$ & $\star$$\star$ & Bachelor & Cancer\\
\hline
\end{tabular}
\caption{A $0.3$-closeness partition}\label{tab:tclose}
\end{table}

Now we are ready to define the main problem studied in this paper.
\begin{problem}
Given an input table $\mathcal{T}$, an \saspace $(\Sigma_s,d)$, and a threshold $t\in [0,1]$,
the \textsc{$t$-Closeness} problem requires to find a $t$-closeness partition of $\mathcal{T}$ with minimum cost.
\end{problem}

Finally we review another two widely-used principles for privacy preserving, namely $k$-anonymity and $l$-diversity, and the combinatorial problems associated with them. A partition is called \emph{$k$-anonymous} if all its \groups have size at least $k$. A (sub-)table $\mathcal{M}$ is called \emph{l-diverse} if at most $|\mathcal{M}|/l$ of the vectors in $\mathcal{M}$ have an identical SA value. A partition is called $l$-diverse if all its \groups are $l$-diverse.
\begin{problem}
Let $\mathcal{T}$ be a table given as input. The \textsc{$k$-Anonymity} ($l$-\textsc{Diversity}) problem requires to find a $k$-anonymous ($l$-diverse) partition of $\mathcal{T}$ with minimum cost.
\end{problem}

\section{NP-hardness Results}\label{sec:nphard}
In this section we study the complexity of the $t$-\textsc{Closeness} problem. The problem is trivial if the given threshold is $t=1$, since putting each vector in a distinct \group produces a 1-closeness partition with cost 0, which is obviously optimal. Our main theorem stated below indicates that this is in fact the only easy case.
\begin{theorem}\label{thm:closeness_nphard}
For any constant $t$ such that $0\leq t<1$, $t$-\textsc{Closeness} is NP-hard.
\end{theorem}

We will prove Theorem~\ref{thm:closeness_nphard} via several reductions, each covering a particular range of $t$, which altogether prove the theorem. We first present a result that relates \tclose to \kanony.

\begin{lemma}\label{lem:reduction}
There is a polynomial-time reduction from \kanony to \tclose with equal-distance space and $t=1-k/n$.
\end{lemma}
\begin{proof}
Let $\mathcal{T}$ be an input table of \kanony. We properly change the SA values of vectors in $T$ to ensure that all
their SA values are distinct; this can be done because the SA values are irrelevant to the objective of the \kanony
problem. Assume w.l.o.g. that the SA values are $\{1,2,\ldots,n\}$. Consider an instance of \tclose with the same
input table $\mathcal{T}$, in which $t=1-k/n$ and the \saspace is the equal-distance space.
The \vect of $\mathcal{T}$ is $(1/n,1/n,\ldots,1/n)$. In the \vect of each size-$r$ \group $\mathcal{T}_r$, there are exactly $r$ coordinates equal to $1/r$ and $n-r$ coordinates equal to $0$. It is easy to see that $\emd(\prob(\mathcal{T}), \prob(\mathcal{T}_r))=(n-r)(1/n)=1-r/n$. Hence, a \group has $t$-closeness if and only if it is of size at least $k$. Therefore, each $k$-anonymous partition of $\mathcal{T}$ is also a $t$-closeness partition, and vice versa. The lemma follows.
\end{proof}

By Lemma~\ref{lem:reduction} we can directly deduce the NP-hardness of \tclose when the threshold $t$ is given as input, using e.g. the NP-hardness of 3-\textsc{Anonymity} \cite{pods04_kanony}. However, to show hardness for constant $t$ that is bounded away from 1, we need $k/n=\Omega(1)$ and thus $k=\Omega(n)$. Unfortunately, the existing hardness results for \kanony only work for $k=O(1)$ and cannot be generalized to large values of $k$. For example, most hardness proofs use reductions from the $k$-dimensional matching problem, but this problem can be solved in polynomial time when $k=\Omega(n)$. Below we show the NP-hardness of \kanony for $k=\Omega(n)$ via reductions different from all previous approaches in the literature.

\begin{theorem}\label{thm:nphard}
For any constant $c$ such that $0<c\leq 1/2$, $(cn)$-\textsc{Anonymity} is NP-hard.
\end{theorem}

To the best of our knowledge, Theorem~\ref{thm:nphard} is the first hardness result for \kanony when $k=\Omega(n)$. We note that the constant $1/2$ is the best possible, since for any $k>n/2$, a $k$-anonymous partition can only contain one \group, namely the table itself. We first prove the following result, which will be used as a starting point in further reductions.

\begin{theorem}\label{thm:nphard_half}
$(n/2)$-\textsc{Anonymity} is NP-hard.
\end{theorem}
\begin{proof}
We will present a polynomial-time reduction from the minimum graph bisection (\textsc{MinBisection}) problem to $(n/2)$-\textsc{Anonymity}. \textsc{MinBisection} is a well-known NP-hard problem \cite{book_npc,bisection} defined as follows: given an undirected graph, find a partition of its vertices into two equal-sized halves so as to minimize the number of edges with exactly one endpoint in each half.

\begin{figure}[t]
\centering
\includegraphics[height=3cm]{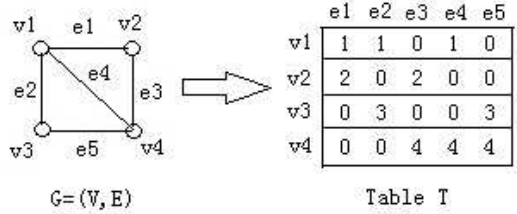}
\caption{Reduction from \textsc{MinBisection} to $(n/2)$-\textsc{Anonymity}.}
\label{fig:bisection}
\end{figure}

Let $G=(V,E)$ be an input graph of \textsc{MinBisection}, where $|V|\geq 4$ is even. Suppose $V=\{v_1,v_2,\ldots,v_n\}$ and $E=\{e_1,e_2,\ldots,e_m\}$. In what follows we construct a table $\mathcal{T}$ of size $n=|V|$ that contains $m=|E|$ \qas. (The \sas are useless in $k$-\textsc{Anonymity} so they will not appear.) This table will serve as the input to the $k$-\textsc{Anonymity} problem with $k=n/2$. Intuitively each row (or vector) of $\mathcal{T}$ corresponds to a vertex in $V$, while each column (or QI) of $\mathcal{T}$ corresponds to an edge in $E$. The alphabet $\Sigma$ is $\{1,2,\ldots,n\}$. For each $i\in [n]$ and $j \in [m]$,\footnote{We use $[q]$ to interchangeably denote $\{1,2,\ldots,q\}$.} let $\mathcal{T}[i][j]=i$ if $v_i \in e_j$, and $\mathcal{T}[i][j]=0$ if $v_i \not\in e_j$. Thus each column contains exactly two non-zero elements corresponding to the two endpoints of the associated edge. See Figure~\ref{fig:bisection} for a toy example. It is easy to see that $\calT$ can be constructed in polynomial time.

Before delving into the reduction, we first prove a result concerning the partition cost of $\mathcal{T}$. Any $(n/2)$-anonymous partition of $\calT$ contains at most two \groups. For the trivial partition that only contains $\calT$ itself, the cost is $n\cdot m$ because all elements in $\calT$ should be suppressed. Thus an $(n/2)$-anonymous partition with minimum cost should consist of exactly two \groups. Suppose $\calP=\{\calT_1,\calT_2\}$ is an $(n/2)$-anonymous partition of $\calT$ where $|\calT_1|=|\calT_2|=n/2$. Let $\{V_1,V_2\}$ be the corresponding partition of $V$ (recall that each vector in $\calT$ corresponds to a vertex in $V$). Consider $Gen(\calT_1)$, the generalization of $\calT_1$. For any column $j \in [m]$, if some endpoint of $e_j$, say $v_i$, belongs to $V_1$, then $\calT[i][j]=i$. By our construction of $\calT$, any other element in the $j$-th column does not equal to $i$. Since $|\calT_1|\geq n/2\geq 2$, column $j$ of $\calT_1$ must be suppressed to $\star$. On the other hand, if none of $e_j$'s endpoints belongs to $V_1$, then column $j$ of $\calT_1$ contains only zeros and thus can stay unsuppressed. Therefore, we obtain
\begin{equation}\label{equ:onepart}
cost(\calT_1)= |\calT_1| \cdot (|E_{11}|+|E_{12}|) = n(|E_{11}|+|E_{12}|)/2,
\end{equation}
where $E_{pq}$ denotes the set of edges with one endpoint in $V_p$ and another in $V_q$, for $p,q\in\{1,2\}$.
Similarly we have $cost(\calT_2)=n(|E_{22}|+|E_{12}|)/2$. Hence the cost of the partition $\calP$ is
\begin{equation}\label{equ:cut_cost}
cost(\calP)=\sum_{p=1}^{2}cost(\calT_p)=n(|E|+|E_{12}|)/2,
\end{equation}
noting that $|E|=|E_{11}|+|E_{12}|+|E_{22}|$.

We now prove the correctness of the reduction. Let $OPT$ be the minimum size of any cut $\{V_1,V_2\}$ of $G$ with $|V_1|=|V_2|$, and $OPT'$ be the minimum cost of any $(n/2)$-anonymous partition of $\calT$. We prove that $OPT'=n(|E|+OPT)/2$, which will complete the reduction from \textsc{MinBisection} to $(n/2)$-\textsc{Anonymity}. Let $\{V_1,V_2\}$ be the cut of $G$ achieving the optimal cut size $OPT$, where $|V_1|=|V_2|=n/2$. Using notation introduced before, we have $OPT=|E_{12}|$. Let $\calP=\{\calT_1,\calT_2\}$ where $\calT_p=\{\calT[i]~|~v_i \in V_p\}$ for $p\in\{1,2\}$. Clearly $\calP$ is an $(n/2)$-anonymous partition of $\calT$. By Equation (\ref{equ:cut_cost}) we have $OPT' \leq cost(\calP)=n(|E|+OPT)/2$.

On the other hand, let $\calP'=\{\calT'_1,\calT'_2\}$ be an $(n/2)$-anonymous partition with $cost(\calP')=OPT'$. We have $|\calT'_1|=|\calT'_2|=n/2$. Consider the partition $\{V'_1,V'_2\}$ of $V$ with $V'_p=\{v_i~|~\calT[i] \in \calT'_p\}$ for $p\in\{1,2\}$. Since $|V'_1|=|V'_2|=n/2$, we have $OPT \leq |E'_{12}|$ where $E'_{12}$ denotes the set of edges with one endpoint in $V'_1$ and another in $V'_2$. By Equation (\ref{equ:cut_cost}) we have $OPT'=n(|E|+|E'_{12}|)/2 \geq n(|E|+OPT)/2$. Combined with the previously obtained inequality $OPT'\leq n(|E|+OPT)/2$, we have shown that $OPT'=n(|E|+OPT)/2$. By the analyses we also know that an optimal $(n/2)$-anonymous partition of $\calT$ can easily be transformed to an optimal equal-sized cut of $G$. This finishes the reduction from \textsc{MinBisection} to $(n/2)$-\textsc{Anonymity}, and completes the proof of Theorem~\ref{thm:nphard_half}.
\end{proof}


\begin{figure}
\centering
\includegraphics[height=3.2cm]{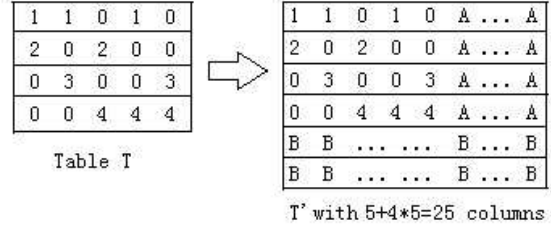}
\caption{Reduction from $(n/2)$-\textsc{Anonymity} to $(n/3)$-\textsc{Anonymity}.}
\label{fig:2}
\end{figure}

\begin{theorem}\label{thm:nphard_part2}
For any constant $c$ such that $0<c\leq 1/3$, $(cn)$-\textsc{Anonymity} is NP-hard.
\end{theorem}
\begin{proof}
Fix $0<c\leq 1/3$. We reduce $(n/2)$-\textsc{Anonymity} to $(cn)$-\textsc{Anonymity}, which will prove the NP-hardness of the latter due to Theorem~\ref{thm:nphard_half}. Let $\calT$ be an instance of $(n/2)$-\textsc{Anonymity} with $n$ rows and $m$ QI columns. Choose two fresh symbols $\lambda_1,\lambda_2$ not appearing in $\calT$. We construct a new table $\calT'$ with $n'=n/2c$ rows and $m'=m+nm$ QI columns as follows\footnote{Here we assume $n/2c$ is an integer, otherwise we can use $\lfloor n/2c\rfloor$ instead and get the same result, with more tedious analyses. Similar issues appear also in other proofs, which we will not mention again.}.
For all $1\leq i\leq n$, let $\calT'[i][j]=\calT[i][j]$ for $1\leq j\leq m$, and $\calT'[i][j]=\lambda_1$ for $m+1\leq j\leq m'$. For all $n+1\leq i\leq n'$ and $1\leq j\leq m'$, let $\calT'[i][j]=\lambda_2$. This finishes the description of $\calT'$. See Figure~\ref{fig:2} for an example where $c=1/3$, $\lambda_1=\textrm{`A'}$, and $\lambda_2=\textrm{`B'}$. Clearly $\calT'$ can be constructed in polynomial time.

Let $OPT$ denote the minimum cost of an $(n/2)$-anonymous partition of $\calT$ and $OPT'$ be the minimum cost of a $(cn')$-anonymous partition of $\calT'$. We next prove $OPT=OPT'$, which will complete the reduction from $(n/2)$-\textsc{Anonymity} to $(cn)$-\textsc{Anonymity}.

On one hand, let $\calP=\{\calT_1,\calT_2\}$ be an $(n/2)$-anonymous partition of $\calT$ with $cost(\calP)=OPT$. We have $|\calT_1|=|\calT_2|=n/2$. Define a partition $\calP'$ of $\calT'$ as $\{\calT'_1,\calT'_2,\calT'_3\}$, where $\calT'_p=\{\calT'[i]~|~\calT[i] \in \calT_p, 1\leq i\leq n\}$ for $p\in\{1,2\}$, and $\calT'_3=\{\calT'[i]~|~n+1\leq i\leq n'\}$. We have $|\calT'_1|=|\calT'_2|=n/2\geq c(n/2c)=cn'$, and $|\calT'_3|/n'=(n/2c-n)/(n/2c)=1-2c \geq c$ as $c \leq 1/3$. Hence $\calP'$ is a $(cn)$-anonymous partition of $\calT'$. It is easy to verify that $cost(\calP')=cost(\calP)=OPT$, implying that $OPT'\leq OPT$.

On the other hand, let $\calP'=\{\calT'_1,\ldots,\calT'_r\}$ be a $(cn)$-anonymous partition of $\calT'$ with $cost(\calP')=OPT'$. For the simplicity of expression, we call $\calT'[i]$ an \emph{old} row if $1\leq i'\leq n$, and call it a \emph{new} row if $n+1\leq i\leq n'$. First assume that there exists $\calT'_p \in \calP'$ that contains both an old row and a new row. By our construction of $\calT'$, an old row and a new row differ in all the last $nm$ coordinates, and thus the cost for generalizing $\calT'_p$ is at least $2nm$. Since $OPT\leq nm$, we have $OPT' > OPT$ in this case, which cannot happen since we already proved $OPT' \leq OPT$. Therefore, for any $\calT'_p \in \calP'$, it either contains only old rows or contains only new rows. Assume w.l.o.g. that $\calT'_1,\ldots,\calT'_{r'}$ are the sub-tables in $\calP'$ that contain only old rows. Since all new rows are identical by our construction, we have
\begin{equation}\label{equ:nphard_part2}
OPT'=cost(\calP')=\sum_{p=1}^{r'}cost(\calT'_p).
\end{equation}
We now define a partition of $\calT$ as $\calP=\{\calT_1,\ldots,\calT_{r'}\}$, where $\calT_p=\{\calT[i]~|~\calT'[i] \in \calT'_p\}$ for all $1\leq p\leq r'$. Because $|\calT_p|=|\calT'_p|\geq cn'=c(n/2c)=n/2$, $\calP$ is an $(n/2)$-anonymous partition of $\calT$. As the last $nm+1$ columns are identical for all old rows of $\calT'$, we have $OPT\leq cost(\calP)=\sum_{p=1}^{r'}cost(\calT_p)=\sum_{p=1}^{r'}cost(\calT'_p)=OPT'$ by Equation (\ref{equ:nphard_part2}). Combined with that $OPT' \leq OPT$ obtained previously, we obtain that $OPT=OPT'$, and that an optimal $(cn)$-anonymous partition of $\calT'$ can be easily transferred to an optimal $(n/2)$-anonymous partition of $\calT$. This finishes the reduction from $(n/2)$-\textsc{Anonymity} to $(cn)$-\textsc{Anonymity}, and completes the proof of Theorem~\ref{thm:nphard_part2}.
\end{proof}

So far we have shown the hardness of $(cn)$-\textsc{Anonymity} for $c \in (0,1/3] \cup \{1/2\}$. For the remaining case $c\in(1/3,1/2)$ we need a different reduction.
\begin{theorem}\label{thm:nphard_part3}
For any constant $c$ such that $1/3<c<1/2$, $(cn)$-\textsc{Anonymity} is NP-hard.
\end{theorem}
\begin{proof}
Fix $1/3<c<1/2$. We will present a polynomial reduction from the following problem to $(cn)$-\textsc{Anonymity}: given an undirected graph $G=(V,E)$, decide whether $G$ contains a clique (i.e., a subgraph in which every pair of vertices have an edge between them) that contains exactly $|V|/2$ vertices. Call this problem \textsc{HalfClique}. The NP-hardness of \textsc{HalfClique} easily follows from that of the well-known maximum clique problem, as can be seen as follows.
We reduce the classical \textsc{Clique} problem to \textsc{HalfClique}. Given a graph $G=(V,E)$ and an integer $k\leq |V|$, the \textsc{Clique} problem asks whether $G$ contains a clique with exactly $k$ vertices. This is a well-known NP-hard problem \cite{book_npc}. Now construct another graph $G'$ based on $G$ as follows: if $k\geq |V|/2$, then add $2k-|V|$ new isolated vertices to $V$; if $k<|V|/2$, then add $|V|-2k$ new vertices to $V$ and connecting them with each other as well as all original vertices in $V$. Let $V'$ be the new vertex set. It is easy to verify that $G$ has a clique of size $k$ if and only if $G'$ has a clique of size $|V'|/2$, which completes the reduction.

Let $G=(V,E)$ be an input graph of \textsc{HalfClique} with $|V|=n\geq 4$ and $|E|=m$. Assume $V=\{v_1,\ldots,v_n\}$ and $E=\{e_1,\ldots,e_m\}$.
We construct a table $\calT$ with $n'=n/2c$ rows and $m$ QI columns as follows. For $1\leq i\leq n$ and $1\leq j\leq m$, let $\calT[i][j]=i$ if $v_i \in e_j$, and $\calT[i][j]=0$ otherwise. For $n+1\leq i\leq n'$ and $1\leq j\leq m$, let $\calT[i][j]=i$. (Note that, in some sense, this construction can be seen as a combination of those used in the proof of Theorems~\ref{thm:nphard_half} and \ref{thm:nphard_part2}; however the analysis will be different and more intriguing.)

We first prove a result regarding the structure of an optimal $(cn)$-partition of $\calT$. Call $\calT[i]$ an old row if $1\leq i\leq n$, and a new row if $n+1\leq i\leq n'$. We assume that $\frac{n}{2c}-n\geq 2$, i.e., $\calT$ contains at least two new rows; this is without loss of generality because $c$ is a constant smaller than $1/2$.
Since $c>1/3$, any $(cn)$-anonymous partition contains at most two \groups. The trivial partition that consists of $\calT$ itself need to suppress every coordinate in the table, because a new row and an old row do not share common values. Therefore, the minimum cost $(cn)$-partition of $\calT$ contains exactly two \groups.

Denote by $OPT$ the minimum cost of a $(cn)$-anonymous partition of $\calT$. We claim that $G$ contains a clique of size $n/2$ if and only if $OPT\leq n'm-(n/2){n/2 \choose 2}$.
First consider the ``only if'' part. Assume $V_2 \subseteq V$ is a clique of size $n/2$, and let $V_1 = V \setminus V_2$. Then $|V_2|=|V_1|=n/2$. For $p,q\in\{1,2\}$, denote by $E_{pq}$ the set of edges with one endpoint in $V_p$ and another in $V_q$. We define a partition $\calP=\{\calT_1,\calT_2\}$ of $\calT$ by letting $\calT_1=\{\calT[i]~|~v_i \in V_1\}$ and $\calT_2 = \calT \setminus \calT'_1$. Since $|\calT_1|=n/2=c(n/2c)=cn'$ and $|\calT_2|/n'=(n/2c-n/2)/(n/2c)=1-c\geq c$, $\calP$ is a $(cn)$-anonymous partition. Similar to the proof of Theorem~\ref{thm:nphard_half}, we have $cost(\calT_1)=n(|E_{11}|+|E_{12}|)/2$ (see Equation (\ref{equ:onepart}) and its proof). Since $\calT_2$ contains both old and new rows, we have $cost(\calT_2)=|\calT_2|\cdot m=(n'-n/2)m$. Therefore,
\begin{eqnarray*}
OPT &\leq& cost(\calP)=cost(\calT_1)+cost(\calT_2)\\
&=&n(|E_{11}|+|E_{12}|)/2+(n'-n/2)m\\
&=&n(m-|E_{22}|)/2+(n'-n/2)m\\
&=&n'm-(n/2)|E_{22}|\\
&=&n'm-(n/2){n/2 \choose 2},
\end{eqnarray*}
where the last equality holds because $V_2$ is a clique of size $n/2$. This proves the ``only if'' part of the claim.

Next we consider the ``if'' direction.
Let $\calP=\{\calT_1,\calT_2\}$ be a $(cn)$-partition with $cost(\calP)=OPT \leq n'm-(n/2){n/2 \choose 2}$. As argued before, every sub-table that contains both old and new rows need to be suppressed totally. Thus, if both $\calT_1$ and $\calT_2$ contain both old and new rows, then $cost(\calP)=n'm$, which is worst possible. In this case we can change $\calT_1$ to be any set of $n/2$ old rows and let $\calT_2 = \calT \setminus \calT_1$ to obtain a $(cn)$-partition with no worse cost. Therefore, in what follows we assume w.l.o.g. that $\calT_1$ consists of only old rows.

Let $V_1=\{v_i~|~\calT[i]\in \calT_1\}$ and $V_2=V\setminus V_1$. Define $E_{pq}$ analogously as before for $p,q\in \{1,2\}$. Similar to the proof of Theorem~\ref{thm:nphard_half}, we have $cost(\calT_1)=|V_1|(|E_{11}|+|E_{12}|)$ (just replace $n/2$ with $|V_1|$ in Equation (\ref{equ:onepart})). Also $cost(\calT_2)=|\calT_2|\cdot m=(n'-|V_1|)m$ since $\calT_2$ contains both old and new rows. Hence, $cost(\calP)=|V_1|(|E_{11}|+|E_{22}|)+(n'-|V_1|)m=|V_1|(m-|E_{22}|)+(n'-|V_1|)m=n'm-|V_1|\cdot |E_{22}|$. On the other hand, $cost(\calP)=OPT \leq n'm-(n/2){n/2 \choose 2}$. Thus we have $|V_1|\cdot |E_{22}| \geq (n/2){n/2 \choose 2}$. As $|V_1|+|V_2|=n$ and $|E_{22}|\leq {|V_2| \choose 2}$, we obtain that
\begin{equation}\label{equ:temp}
(n-|V_2|){|V_2| \choose 2} \geq |V_1| \cdot |E_{22}| \geq \frac{n}{2}{n/2 \choose 2}.
\end{equation}

Because $|V_1|=|\calT_1|\geq cn'=n/2$, we have $|V_2|\leq n/2$. Define a fucntion $f: [0,n/2] \rightarrow \mathbb{R}$ as $f(x)=(n-x){x\choose 2}=(n-x)x(x-1)/2$ for all $0\leq x\leq n/2$. Then Equation (\ref{equ:temp}) indicates that $f(|V_2|) \geq f(n/2)$. Since $f(0)=0$, $|V_2|\geq 1$ holds. Let $f'$ be the derivative of $f$ with respect to $x$. It is easy to verify that $f'(x)=\frac{1}{2}(-3x^2+2(n+1)x-n)$. The minimum value of $f'(x)$ when $1\leq x\leq n/2$ can only be obtained at $x \in \{1,n/2,(n+1)/3\}$. Simple calculations show that $f'(1),f'(n/2)$, and $f'((n+1)/3)$ are all positive. Hence $f'(x)>0$ for all $1\leq x\leq n/2$, which means that $f(x)$ is strictly monotone increasing on $[1,n/2]$. Since we know that $f(|V_2|) \geq f(n/2)$ and that $1\leq |V_2| \leq n/2$, it must hold that $|V_2|=n/2$. Therefore (\ref{equ:temp}) holds with two equalities. We thus have $|E_{22}|={n/2\choose 2}$, implying that $V_2$ is a clique of size $n/2$.

We have shown that $G$ has a clique of size $n/2$ if and only if $\calT$ has a $(cn)$-anonymous partition of cost at most $n'm-(n/2){n\choose 2}$. This completes the reduction from \textsc{HalfClique} to $(cn)$-\textsc{Anonymity}, from which Theorem~\ref{thm:nphard_part3} follows.
\end{proof}

Now Theorem~\ref{thm:nphard} follows straightforward from Theorems~\ref{thm:nphard_half}, \ref{thm:nphard_part2} and \ref{thm:nphard_part3}. Interestingly, the three reductions work for disjoint ranges of $c$, which altogether give the desired result.
By Lemma~\ref{lem:reduction} we obtain:
\begin{corollary}\label{cor:closeness_atleasthalf}
For any constant $t$ such that $1/2 \leq t < 1$, $t$-\textsc{Closeness} is NP-hard even with equal-distance space.
\end{corollary}

We next show the hardness of $t$-\textsc{Closeness} for $0\leq t<1/2$ by two different reductions from the 3-dimensional matching problem, each of which covers a different range of $t$.

\begin{theorem}\label{thm:closeness_less13}
For any constant $t$ such that $0\leq t<1/3$, $t$-\textsc{Closeness} is NP-hard even if $|\Sigma_s|=3$.
\end{theorem}
\begin{proof}
Fix $0\leq t<1/2$. We perform a polynomial-time reduction from the 3-dimensional matching problem (\textsc{3D-Matching}) to $t$-\textsc{Closeness}. The input of \textsc{3D-Matching} consists of three equal-sized pairwise-disjoint sets $X,Y$, and $Z$, together with a collection $S$ of 3-tuples from $X \times Y \times Z$. The goal is to decide whether there exists a set of $|X|$ tuples from $S$ that covers each element of $X \cup Y \cup Z$ exactly once. This problem is well known to be NP-hard \cite{book_npc}.

Consider an instance of \textsc{3D-Matching}. Assume $|X|=|Y|=|Z|=n$, $U=X \cup Y \cup Z=\{v_1,v_2,\ldots,v_{3n}$, and the set of tuples is $S=\{e_1,e_2,\ldots,e_m\}$. Each tuple in $S$ is regarded as a subset of $U$ of size 3. The reduction that we will use is similar to that in \cite{edbt10_ldiversity}. We construct an instance of $t$-\textsc{Closeness} as follows. The table $\calT$ has $3n$ rows and $m$ QI columns as well as an SA column. For every $1\leq i\leq 3n$ and $1\leq j\leq m$, let $\calT[i][j]=i$ if $v_i \not\in e_j$ and $\calT[i][j]=0$ if $v_i \in e_j$. Let $\calT[i][m+1]$ be 1, 2, or 3, if $v_i$ belongs to $X$, $Y$, or $Z$, respectively. The \saspace is the equal-distance space. Notice that each QI column of $\calT$ contains exactly three zeros, corresponding to the three elements in the tuple associated with this column. Also note that the \vect of $\calT$ is $\prob(\calT)=(1/3,1/3,1/3)$ since $|X|=|Y|=|Z|$.

We will prove that, there exists $n$ tuples of $S$ whose union equals $U=X \cup Y \cup Z$ if and only if $\calT$ has a $t$-closeness partition of cost at most $3n(m-1)$. This will complete the reduction from \textsc{3D-Matching} to $t$-\textsc{Closeness}.

First consider the ``only of'' direction. Assume that there exists $S' \subseteq S$, $|S'|=n$, such that $\bigcup_{e\in S'}e = U$. We assume w.l.o.g. that $S'=\{e_1,e_2,\ldots,e_n\}$. Define a partition $\calP$ of $\calT$ as follows: $\calP=\{\calT_1,\ldots,\calT_n\}$, where $\calT_p=\{\calT[i]~|~v_i \in e_p\}$ for all $1\leq p\leq n$. Clearly $|\calT_1|=\ldots=|\calT_n|=3$. Since each $e_p$ contains exactly one element from each of $X$, $Y$ and $Z$, we have $\prob(\calT_p)=(1/3,1/3,1/3)=\prob(\calT)$. Hence $\calP$ is a $t$-closeness (and in fact 0-closeness) partition of $\calT$. By our construction, for each $p\in [n]$, the $p$-th column of $\calT_p$ consists of three zeros, and every other column contains at least two different QI values. Thus $cost(\calT_p)=3(m-1)$, and $cost(\calP)=\sum_{p=1}^{n}cost(\calT_p)=3n(m-1)$. The ``only if'' direction is proved.

We next consider the ``if'' direction. Let $\calP=\{\calT_1,\ldots,\calT_r\}$ be a $t$-closeness partition of $\calT$ with cost at most $3n(m-1)$. We claim that $|\calT_p|\geq 3$ for all $p\in [r]$. Assume to the contrary that $|\calT_p| \leq 2$ for some $p$. Then $\prob(\calT_p)$ is either $(0,1/2,1/2)$ or $(0,0,1)$ up to permutations of the coordinates. It is easy to verify that $\emd(\prob(\calT_p),\prob(\calT))\geq 1/3>c$ in both cases, which contradicts the fact that $\calP$ is a $t$-closeness partition. Hence, $|\calT_p| \geq 3$. If $|\calT_p| \geq 4$, then $cost(\calT_p)=|\calT_p|\cdot m$, because each column of $\calT$ consists of three zeros and $3n-3$ distinct non-zero values and thus needs to be suppressed entirely in $\calT_p$. If $|\calT_p|=3$, then $cost(\calT_p)=3(m-1)$ if there is a tuple in $S$ that contains the three elements associated with the vectors in $\calT_p$ (in which case the column corresponding to this tuple needs not be suppressed), and $cost(\calT_p)=3m$ otherwise. Since $cost(\calP)=3n(m-1)$, every \group $\calT_p$ is of size 3 and induces a tuple, say $e_p \in S$. Then $\{e_1,\ldots,e_p\}$ is a set of $n$ tuples whose union equals $U$, proving the ``if'' direction.
This completes the reduction from \textsc{3D-Matching} to $t$-\textsc{Closeness}, and  Theorem~\ref{thm:closeness_less13} follows.
\end{proof}

Finally we come to the last part $t \in [1/3,1/2)$.
\begin{theorem}\label{thm:closeness_last}
For any constant $t$ such that $1/3\leq t<1/2$, $t$-\textsc{Closeness} is NP-hard even if $|\Sigma_s|=4$.
\end{theorem}
\begin{proof}
Fix $1/3\leq t<1/2$. We give a reduction from \textsc{3D-Matching} to $t$-\textsc{Closeness} similar to that used in the the proof of Theorem~\ref{thm:closeness_less13}, with some more ingredients. Consider an instance of \textsc{3D-Matching}. The element set is $U=X \cup Y \cup Z=\{v_1,v_2,\ldots,v_{3n}\}$ where $|X|=|Y|=|Z|=n$. The tuple set is $S=\{e_1,\ldots,e_m\}$ where each $e_i$, $1\leq i\leq m$, is a subset of $U$ of size 3 that consists of exactly one element from each of $X$, $Y$, and $Z$. The goal is to decide whether there exists $S' \subseteq S$, $|S'|=n$, such that $\bigcup_{e\in S'}e = U$.

We set up an instance of $t$-\textsc{Closeness} as follows. The table $\calT$ consists of $n'=3n/(1-2t)$ rows, $m$ QI columns, and an SA column. For all $1\leq i\leq 3n$ and $1\leq j\leq m$, $\calT[i][j]=i$ if $v_i \not\in e_j$ and $\calT[i][j]=0$ if $v_i\in e_j$. For $1\leq i\leq 3n$, $\calT[i][m+1]=1$ if $v_i\in X$, $\calT[i][m+1]=2$ if $v_i \in Y$, and $\calT[i][m+1]=3$ if $v_i \in Z$. For $3n+1\leq i\leq n'$, $\calT[i][j]=i$ for $1\leq j\leq m$, and $\calT[i][m+1]=4$. Note that $\Sigma_s=\{1,2,3,4\}$. Define the distance function of the \saspace as $d(1,2)=d(1,3)=d(2,3)=1$ and $d(4,1)=d(4,2)=d(4,3)=1/2$; this clearly forms a metric on $\Sigma_s$. It is easy to verify that $\prob(\calT)=(\frac{1-2t}{3},\frac{1-2t}{3},\frac{1-2t}{3},2t)$. The goal is to decide whether $\calT$ has a $t$-closeness partition.
Before showing the correctness of the reduction, we present a formula for computing the EMD between two distributions under this metric. Let $\textbf{A}=(a_1,a_2,a_3,a_4)$ and $\textbf{B}=(b_1,b_2,b_3,b_4)$ be two \vects with $a_4\geq b_4$. Then,
\begin{equation}\label{equ:special_metric}
\emd(\textbf{A},\textbf{B})=\frac{1}{2}(a_4-b_4)+\sum_{i\in\{1,2,3\}:a_i\geq b_i}(a_i-b_i).
\end{equation}
This can be seen as follows. Let $S_{\geq}=\{1\leq i\leq 4~|~a_i\geq b_i\}$ and $S_{<}=\{1,2,3,4\}\setminus S_{\geq}$. We have $4\in S_{\geq}$ and $\sum_{i\in S_{\geq}}(a_i-b_i)=\sum_{j\in S_{<}}(b_j-a_j)$.
To transform $\mathbf{A}$ to $\mathbf{B}$, we need to move $M=\sum_{i\in S_{\ge}}(a_i-b_i)$ amount of mass from $S_{\ge}$ to $S_<$. $a_4-b_4$ amount of mass at point 4 can be moved out by distance $1/2$, while the remaining amount must be moved by distance 1. Therefore $\emd(\mathbf{A},\mathbf{B})= \frac{1}{2}(a_4-b_4)+\sum_{i \in S_{\ge}\setminus \{4\}}(a_i-b_i)=\frac{1}{2}(a_4-b_4)+\sum_{i\in\{1,2,3\}:a_i\geq b_i}(a_i-b_i)$.

We prove that the answer to the matching instance is yes if and only if $\calT$ has a $t$-closeness partition of cost at most $3n(m-1)$. First consider the ``only if'' direction. Assume w.l.o.g. that $S'=\{e_1,\ldots,e_n\}$ satisfies $\bigcup_{e\in S'}e=U$. Define a partition $\calP=\{\calT_1,\calT_2,\ldots,\calT_{n}\} \cup \{\calT'_{3n+1},\calT'_{3n+2},\ldots,\calT'_{n'}\}$, where $\calT_p=\{\calT[i]~|~i\in e_p\}$ for $1\leq p\leq n$ and $\calT'_{p}=\{\calT[p]\}$ for $3n+1\leq p\leq n'$, i.e., each $\calT'_p$ consists of a single row. By similar arguments as in the proof of Theorem~\ref{thm:closeness_less13}, $cost(\calT_{p})=3(m-1)$ for $1\leq p\leq n$, and obviously $cost(\calT'_{p})=0$ for $3n+1\leq p\leq n'$. Hence $cost(\calP)=3n(m-1)$. It remains to show that $\calP$ is a $t$-closeness partition. For $1\leq p\leq n$, $\prob(\calT_p)=(1/3,1/3,1/3,0)$, by Equation~(\ref{equ:special_metric}) we have $\emd(\prob(\calT_p),\prob(\calT))=\frac{1}{2}\cdot 2t=t$. Since $\prob(\calT'_{q})=(0,0,0,1)$ for all $3n+1\leq q\leq n'$, $\emd(\prob(\calT'_{q}),\prob(\calT))=\frac{1}{2}(1-2t)\leq t$ as $t\geq 1/3$ (actually this holds for all $t\geq 1/4$). This proves that $\calP$ is a $t$-closeness partition, and hence the ``only if'' direction.

Now consider the ``if'' direction. Let $\calP=\{\calT_1,\ldots,\calT_r\}$ be a $t$-closeness partition of $\calT$ with cost at most $3n(m-1)$. Call $\calT[i]$ an old row if $1\leq i\leq 3n$, and a new row if $i>3n$. By our construction of $\calT$, it is clear that $cost(\calT_p)=|\calT_p|\cdot m$ if $|\calT_p|\geq 2$ and $\calT_p$ contains at least one new row. Now let $\calT_p$ be a \group containing only old rows. If $|\calT_p|\leq 2$, then $\prob(\calT_p))$ is equivalent to $(1,0,0,0)$ or $(1/2,1/2,0,0)$ up to permutations of the first three coordinates. By (\ref{equ:special_metric}) and the fact that $t<1/2$, we can verify that $\emd(\prob(\calT_p),\prob(\calT))>t$ in both cases. Therefore $|\calT_p|\geq 3$. Analogous to the proof of Theorem~\ref{thm:closeness_less13}, we know that $cost(\calT_p)=3(m-1)$ if $\calT_p$ consists of three old rows corresponding to three elements in the same tuple, and $cost(\calT_p)=|\calT_p|\cdot m$ otherwise. Thus for $cost(\calP)=3n(m-1)$ it must be the case that there exist $n$ \groups each of which consists of three old rows, and each of the remaining \groups consists of exact one new row. As \groups are disjoint, they together cover all the $3n$ old rows, which naturally induces $n$ tuples of $S$ whose union equals $U$. The ``if'' direction is thus proved.
This completes the reduction from \textsc{3D-Matching} to $t$-\textsc{Closeness}, and Theorem~\ref{thm:closeness_last} follows.
\end{proof}

\section{Exact and Fixed-Parameter Algorithms}\label{sec:exact}
In this section we design exact algorithms for solving $t$-\textsc{Closeness}. Notice that the size of an instance of $t$-\textsc{Closeness} is polynomial in $n$ and $m+1$. The brute-force approach that examines each possible partition of the table to find the optimal solution takes $n^{O(n)}m^{O(1)}=2^{O(n\log n)}m^{O(1)}$ time. We first improve this bound to single exponential in $n$. (Note that it cannot be improved to polynomial unless $\mathrm{P}=\mathrm{NP}$.)
\begin{theorem}
The $t$-\textsc{Closeness} problem can be solved in $2^{O(n)}\cdot O(m)$ time.
\end{theorem}
\begin{proof}
Consider an input table $\calT$ of the $t$-\textsc{Closeness} problem. Assume that $\calP=\{\calT_1,\ldots,\calT_r\}$ is an optimal $t$-closeness partition of $\calT$ (note that we do not know $\calP$; it is only used for analysis).
Obviously there is at most one \group $\calT_p$ with $|\calT_p|>n/2$. We claim that, if $|\calT_p|\leq n/2$ for all $p\in [r]$, then there is a disjoint partition $(A_1,A_2)$ of $\{1,2,\ldots,r\}$ such that $n/4 \leq |\bigcup_{p\in A_i}\calT_p|\leq 3n/4$ for any $i\in\{1,2\}$. This can be seen as follows. Denote by $n(A)=|\bigcup_{p\in A}\calT_p|$ for any $A\subseteq [r]$. Let $(A_1,A_2)$ be the partition of $[r]$ that minimizes $|n(A_1)-n(A_2)|$. Assume w.l.o.g. that $n(A_1) \leq n(A_2)$. If $n(A_2) \leq 3n/4$, the claim is proved. Otherwise, $A_2$ contains at least two \groups, and we move an arbitrary \group from $A_2$ to $A_1$ resulting in a new partition $(A'_1,A'_2)$.
If $n(A'_1)\leq n(A'_2)$, then $|n(A'_1)-n(A'_2)|=n(A'_2)-n(A'_1)<n(A_2)-n(A_1)=|n(A_1)-n(A_2)|$, which contradicts the way in which $(A_1,A_2)$ is chosen. We thus have $n(A'_1)>n(A'_2)$, and so $n(A'_1)\geq n/2$. Since each group has size at most $n/2$, we have $n/2\leq n(A'_1)\leq n(A_1)+n/2<3n/4$, and hence $n/2\geq n(A'_2)=n-n(A'_1)>n/4$. This proves the claim.

For any $M\subseteq \calT$, let $OPT(M)$ denote the minimum cost of any partition of $M$ in which each \group is $t$-close to $\calT$; thus the optimal cost of the problem is $OPT(\calT)$. We now have a natural recursive algorithm for computing $OPT(\calT)$: Enumerate all $\calT_1 \subseteq \calT$ with $n/4\leq |\calT_1|\leq 3n/4$ and find the one minimizing $OPT(\calT_1)+OPT(\calT\setminus\calT_1)$; denote this minimum cost by $OPT_1$. We also exhaustively find $\calT'_1 \subseteq \calT$ with $|\calT'_1|>n/2$ that minimizes $OPT(\calT'_1)+OPT(\calT\setminus\calT'_1)$, which is denoted by $OPT_2$. By our previous analysis, $OPT(\calT)=\min\{OPT_1,OPT_2\}$ and thus we can solve $t$-\textsc{Closeness} by taking the better solution. Two notes on the recursive steps: (1) If we have a table of constant size (say, less than 10) then we can directly solve it in $O(m)$ time by the brute-force approach. (2) If we have a table $\calT'$ such that $\emd(\prob(\calT'),\prob(\calT))>t$ then we return with cost $+\infty$.

We now analyze the running time of the algorithm. Let $f(s)$ denote the running time on a sub-table of $\calT$ of size $s$. When $s\leq 10$ we have $f(s)=O(m)$, and when $s>10$,
\begin{eqnarray*}
f(s)&\leq& \sum_{i=s/4}^{3s/4}{s \choose i}\cdot 2f(3s/4)+\sum_{i=s/2}^{s}{s\choose i}f(s/2)+O(2^s)\\
&\leq& 2^{s+2}f(3s/4)+O(2^s).
\end{eqnarray*}
In the first inequality, the first term stands for the time of enumerating $\calT_1$ with $n/4\leq |\calT_1|\leq 3n/4$, the second term is for the enumeration of $\calT'_1$ with $|\calT'_1|>n/2$, and the third term is responsible for other works such as recording the subsets.
It is easy to verify that this recursion gives $f(n)\leq 2^{O(n)}\cdot O(m)$.
\end{proof}

In many real applications, there are usually only a small number of attributes and distinct attribute values. Thus it is interesting to see whether $t$-\textsc{Closeness} can be solved more efficiently when $m$ and $|\Sigma|$ is small. We answer this question affirmatively in terms of fixed-parameter tractability.
\begin{theorem}
$t$-\textsc{Closeness} is fixed-parameter tractable when parameterized by $m$ and $|\Sigma|$. Thus we can solve $t$-\textsc{Closeness} optimally in polynomial time when $m$ and $|\Sigma|$ are constants.
\end{theorem}
\begin{proof}
Consider an input table $\calT$ with $n$ rows and $m+1$ columns (of which $m$ are QIs and one is SA). For $v\in \Sigma^m$ and $s\in \Sigma_s$, denote by $R_{v,s}$ the set of vectors in $\calT$ that is identical to $(v,s)$, and let $r_{v,s}=|R_{v,s}|$. We thus have $\sum_{v\in \Sigma^m,s\in \Sigma_s}r_{v,s}=n$. We write a integer linear program to characterize the minimum cost of a $t$-closeness partition of $\calT$. For every $v\in \Sigma^m$ and $s\in \Sigma_s$ such that $R_{v,s}\neq \emptyset$, and every $v^* \in (\Sigma\cup\{\star\})^m$ that generalizes $v$, there is a nonnegative integer variable $x(v^*,v,s)$ which means the number of vectors in $R_{v,s}$ that is generalized to $(v^*,s)$ in the partition. We clearly have
\begin{equation}\label{lp:1}
\sum_{v^*:v^* \textrm{~generalizes~} v}x(v^*,v,s)=r_{v,s}, ~~\forall (v,s)\textrm{~s.t.~}R_{v,s}\neq \emptyset.
\end{equation}
Each $v^*\in (\Sigma\cup\{\star\})^m$ induces a \group, denoted $G_{v^*}$, which consists of all vectors whose QI values are generalized to $v^*$. Those \groups together form a partition (note that some \group may be empty). Denoting by $C_{v^*}$ the number of `$\star$'s in $v^*$, the cost of the partition is precisely $\sum_{v^*,v,s}C_{v^*}\cdot x(v^*,v,s)$. Thus the objective function is
\begin{equation}\label{lp:obj}
\textrm{Minimize~}\sum_{v^*,v,s}C_{v^*}\cdot x(v^*,v,s)\;.
\end{equation}
We still need other constraints to ensure that each \group $G_{v^*}$ either is empty or has $t$-closeness.
We do this by adding a set of constraints, for every $v^*$, that characterizes the transportation between the \vects of $G_{v^*}$ and $\calT$ as in the definition of EMD.
First assume that $G_{v^*}$ is non-empty.
We have $|G_{v^*}|=\sum_{v\in\Sigma^m,s\in \Sigma_s}x(v^*,v,s)$.
The probability mass of $i\in \Sigma_s$ in $\prob(G_{v^*})$ is
$\sum_{v\in\Sigma^m}x(v^*,v,i)/|G_{v^*}|$, and that in $\prob(\calT)$ is
$\sum_{v\in\Sigma^m}r_{v,i}/n$. For $i,j\in \Sigma_s$, let $f(v^*,i,j)$ denote the amount of mass moved from $i$ to $j$
in order to transform $\prob(G_{v^*})$ to $\prob(\calT)$. Let $d_{i,j}$ be the distance between $i$ and $j$ in the \saspace. To guarantee the $t$-closeness of $G_{v^*}$ we can write the following constraints:
\begin{eqnarray*}
\sum_{j\in \Sigma_s}f(v^*,i,j)&=&\sum_{v\in\Sigma^m}x(v^*,v,i)/|G_{v^*}|,~\forall i\in \Sigma_s\\
\sum_{i\in \Sigma_s}f(v^*,i,j)&=&\sum_{v\in\Sigma^m}r_{v,j}/n,~\forall j\in\Sigma_s\\
\sum_{i,j\in\Sigma_s}d_{i,j}\cdot f(v^*,i,j)&\leq& t\\
f(v^*,i,j)&\geq& 0,~\forall i,j\in \Sigma_s\;.
\end{eqnarray*}
The first constraint above is not linear. To overcome this, we define $g(v^*,i,j)=f(v^*,i,j)\cdot |G_{v^*}|$, substitute $g(v^*,i,j)$
for $f(v^*,i,j)$ in the above constraints, and expand $|G_{v^*}|$. This produces the following equivalent constraints:
\begin{eqnarray*}
\sum_{j\in \Sigma_s}g(v^*,i,j)&=&\sum_{v\in\Sigma^m}x(v^*,v,i),~\forall i\in \Sigma_s\\
n\sum_{i\in \Sigma_s}g(v^*,i,j)&=&\sum_{v\in\Sigma^m}r_{v,j}\sum_{v\in\Sigma^m,s\in \Sigma_s}x(v^*,v,s),\forall j\in\Sigma_s\\
\sum_{i,j\in\Sigma_s}d_{i,j}\cdot g(v^*,i,j)&\leq& t\cdot\sum_{v\in\Sigma^m,s\in \Sigma_s}x(v^*,v,s)\\
g(v^*,i,j)&\geq& 0,~\forall i,j\in \Sigma_s\;.
\end{eqnarray*}
Note that these constraints hold even if $G_{v^*}$ is empty. Thus they force \group $G_{v^*}$ to be $t$-closeness or empty.
The set of such constraints for all $v^*$, together with (\ref{lp:1}) and (\ref{lp:obj}), compose a \emph{mixed} integer linear program (i.e., only some of the variables
are required to take integer values)
that precisely characterizes the $t$-\textsc{Closeness}
problem on $\calT$.\footnote{A technical issue here is that, in order to apply results for mixed integer linear program, $t$ needs to be a rational number. Nevertheless, for irrational $t$ we can use rationals to approximate the value of $t$ to an arbitrary precision.} The number of variables in the program is $N\leq |\Sigma|^m(|\Sigma|+1)^m|\Sigma_s|+(|\Sigma|+1)^m|\Sigma_s|^2
\leq 2(|\Sigma|+1)^{2m+1}$. The time spent on constructing and writing down this linear program is polynomial in $n, m$, and $N$.
By the result in \cite{mixed_ilp} (Section 5 of it deals with mixed ILP), a mixed linear integer program
with $N$ variables can be solved in $N^{O(N)}L$ time, where $L$ is the number of bits used to encode the program. In our case $L$ is polynomial in $n$ and $m$.
Therefore, we can solve this program, and hence solve $t$-\textsc{Closeness}, in $h(m,|\Sigma|)n^{O(1)}$ time for some function $h$.
This shows that $t$-\textsc{Closeness} is fixed-parameter tractable when parameterized by $m$ and $|\Sigma|$.
\end{proof}

\section{Approximation Algorithm for $k$-Anonymity}\label{sec:kanony}
In this section we give a polynomial-time $m$-approximation algorithm for $k$-\textsc{Anonymity}, which improves the previous best ratio $O(k)$ \cite{icdt05_kanony} and $O(\log k)$ \cite{sigmod07_apx_kanony} when $k$ is relatively large compared with $m$. (We note that the $O(\log k)$-approximation algorithm given in \cite{sigmod07_apx_kanony} is not guaranteed to run in polynomial time for super-constant $k$, while our result holds for all $k$.)
\begin{theorem}\label{thm:kanony_apx}
$k$-\textsc{Anonymity} can be approximated within factor $m$ in polynomial time.
\end{theorem}
\begin{proof}
Consider a table $\calT$ with $n$ rows and $m$ QI columns. Denote by $OPT$ the minimum cost of any $k$-anonymous partition of $\calT$. Partition $\calT$ into ``equivalence classes'' $C_1,\ldots,C_R$ in the following sense: any two vectors in the same class are identical, i.e., they have the same value on each attribute, while any two vectors from different classes differ on at least one attribute. Assume $|C_1| \leq |C_2|\leq \ldots \leq |C_R|$. If $|C_1|\geq k$, then these classes form a $k$-anonymous partition with cost 0, which is surely optimal. Thus we assume $|C_1|<k$, and let $L\in[R]$ be the maximum integer for which $|C_L|<k$. Then $|C_{L'}|\geq k$ for all $L<L'\leq R$.
It is clear that each vector in $C_1\cup \ldots \cup C_L$ contributes at least one to the cost of any partition of $\calT$. Thus $OPT\geq \sum_{i=1}^{L}|C_i|$.

\textbf{Case 1:} $\sum_{i=1}^{L}|C_i|\geq k$. In this case we partition $\calT$ into $R-L+1$ \groups: $\{C_1\cup\ldots \cup C_L, C_{L+1},C_{L+2},\ldots, C_R\}$. This is a $k$-anonymous partition of cost at most $m\cdot \sum_{i=1}^{L}|C_i|\leq m\cdot OPT$.

\textbf{Case 2:} $\sum_{i=1}^{L}|C_i|<k$ and $\sum_{i=1}^{L}|C_i|+\sum_{i=L+1}^{R}(|C_i|-k)\geq k$. We choose $C'_i\subseteq C_i$ for $L+1\leq i\leq R$ satisfying that $|C_i\setminus C'_i|\geq k$ and $\sum_{i=1}^{L}|C_i|+\sum_{i=L+1}^{R}|C'_i|=k$. This can be done because of the second condition of this case. We partition $\calT$ into $R-L+1$ \groups: $\{\bigcup_{i=1}^{L}C_i \cup \bigcup_{i=L+1}^{R}C'_i, C_{L+1}\setminus C'_{L+1},\ldots, C_R\setminus C'_R\}$. This is a $k$-anonymous partition of cost at most $m\cdot k \leq m\cdot OPT$, since $OPT\geq k$.

\textbf{Case 3:} $\sum_{i=1}^{L}|C_i|+\sum_{i=L+1}^{R}(|C_i|-k)<k$. We claim that there exists $i \in\{L+1,\ldots,R\}$ such that any vector in $C_i$ contributes at least one to the cost of any $k$-anonymous partition. Assume the contrary. Then there exists a $k$-anonymous partition such that, for every $L+1\leq i\leq R$, there is a vector $v \in C_i$ whose suppression cost is 0, which means that $v$ belongs to a \group that only contains vectors in $C_i$; denote this \group by $C'_i$. We also know that there is at least one \group in the partition that has positive cost. However, by removing all $C'_i$, $L+1\leq i\leq R$, from $\calT$, the number of vectors left is at most
$n-k(R-L)=\sum_{i=1}^{R}|C_i|-k(R-L)=\sum_{i=1}^{L}|C_i|+\sum_{i=L+1}^{R}(|C_i|-k)<k$,
due to the condition of this case. This contradicts with the property of $k$-anonymous partitions. Therefore the claim holds, i.e., there exists $j\in\{L+1,\ldots,R\}$ such that any vector in $C_j$ contributes at least one to the partition cost. Thus we have $OPT\geq \sum_{i=1}^{L}|C_i|+|C_j|\geq \sum_{i=1}^{L+1}|C_i|$. We partition $\calT$ into $R-L$ \groups: $\{\bigcup_{i=1}^{L+1}C_i,C_{L+2},\ldots,C_{R}\}$. This is a $k$-anonymous partition with cost at most $m\cdot \sum_{i=1}^{L+1}|C_i|\leq m\cdot OPT$.

By the above case analyses, we can always find in polynomial time a $k$-anonymous partition of $\calT$ with cost at most $m \cdot OPT$. This completes the proof of Theorem~\ref{thm:kanony_apx}.
\end{proof}

We note that Theorem~\ref{thm:kanony_apx} implies that $k$-\textsc{Anonymity} can be solved optimally in polynomial time when $m=1$.
This is in contrast to $l$-\textsc{Diversity}, which remains NP-hard when $m=1$ (with unbounded $l$) \cite{tcs12_ldiversity}.

\section{Algorithm for $2$-Diversity}\label{sec:2diversity}
In this part we give the first polynomial time algorithm for solving 2-\textsc{Diversity}.
Let $\calT$ be an input table of $2$-\textsc{Diversity}. The following lemma is crucial to our algorithm.
\begin{lemma}\label{lem:2diversity}
There is an optimal 2-diverse partition of $\calT$ in which every \group consists of 2 or 3 vectors with distinct SA values.
\end{lemma}
\begin{proof}
It suffices to show that any 2-diverse sub-table $M \subseteq \calT$ can be further partitioned into \groups each of which consists of 2 or 3 vectors with distinct SA values (note that partitioning a \group does not increase the generalization cost). We use induction on the size of $M$. When $|M|=2$ or $3$ it can be verified directly. Now consider $M \subseteq \calT$ of size $t\geq 4$. Suppose $M$ contains $k$ SA values $\{1,2,\ldots,k\}$ where $k\geq 2$. Let $a_i$ be the number of vectors in $M$ with SA value $i$, for $i\in [k]$. Assume w.l.o.g. that $a_1\geq a_2\geq \ldots \geq a_k$. Let $A_1$ and $A_2$ be two vectors with SA value 1 and 2, respectively. Partition $M$ into $\{A_1,A_2\}$ and $M'=M\setminus \{A_1,A_2\}$. We only need to show that $M'$ is 2-diverse, so that we can use induction on it. We perform a case analysis as follows.
\begin{itemize}
\item $a_1=1$. Then $M'$ consists of at least two vectors with distinct SA values, and thus is 2-diverse.
\item $k=2$. Since $M$ is 2-diverse, we have $a_1=a_2$. Then $M'$ still contains the same number of SA values 1 and 2, so it remains 2-diverse.
\item $a_1\geq 2,k\geq 3,a_1>a_3$. The highest frequency of any SA value in $M'$ is $a_1-1 \leq |M|/2-1=|M'|/2$, and thus $M'$ is 2-diverse.
\item $a_1=a_3\geq 2,k\geq 3$. In this case $|M|\geq 3a_3$. The highest frequency of an SA value in $M'$ is $a_3$. We have $|M'|-2a_3=|M|-2-2a_3\geq a_3-2\geq 0$, so $M'$ is 2-diverse.
\end{itemize}
All the cases are covered above and hence Lemma~\ref{lem:2diversity} is proved.
\end{proof}

Giving Lemma~\ref{lem:2diversity}, the rest of the proof is basically the same with that of the polynomial-time tractability of 2-\textsc{Anonymity} given in \cite{icalp10_kanony}. We restate the proof for completeness.
We reduce 2-\textsc{Diversity} to a combinatorial problem called \textsc{Simplex Matching} introduced in \cite{stoc07_matching}, which admits a polynomial algorithm \cite{stoc07_matching}. The input of \textsc{Simplex Matching} is a hypergraph $H=(V,E)$ containing edges of sizes 2 and 3 with nonnegative edge costs $c(e)$ for all edges $e\in E$. In addition $H$ is guaranteed to satisfy the following \emph{simplex condition}: if $\{v_1,v_2,v_3\}\in E$, then $\{v_1,v_2\},\{v_2,v_3\},\{v_3,v_1\}$ are also in $E$, and $c(\{v_1,v_2\})+c(\{v_2,v_3\})+c(\{v_1,v_3\})\leq 2\cdot c(\{v_1,v_2,v_3\})$. The goal is to find a perfect matching of $H$ (i.e., a set of edges that cover every vertex $v\in V$ exactly once) with minimum cost (which is the sum of costs of all chosen edges).

Let $\calT$ be an input table of 2-\textsc{Diversity}. We construct a hypergraph $H=(V,E)$ as follows. Let $V=\{v_1,v_2,\ldots,v_n\}$ where $v_i$ corresponds to the vector $\calT[i]$. For every two vectors $\calT[i],\calT[j]$ (or three vectors $\calT[i],\calT[j],\calT[k]$) with distinct SA values, there is an edge $e=\{v_i,v_j\}$ (or $e=\{v_i,v_j,v_k\})$ with cost equal to $cost(\{\calT[i],\calT[j]\})$ (or $cost(\{\calT[i],\calT[j],\calT[k]\})$). Consider any 3D edge $e=\{v_i,v_j,v_k\}$. Since each column that needs to be suppressed in $\{\calT[i],\calT[j]\}$ must also be suppressed in $\{\calT[i],\calT[j],\calT[k]\}$, we have $c(e)/3\geq c(\{v_i,v_j\})/2$. Similarly, $c(e)/3\geq c(\{v_i,v_k\})/2$ and $c(e)/3\geq c(\{v_j,v_k\})/2$. Summing the inequalities up gives $2c(e)\geq c(\{v_i,v_j\})+c(\{v_i,v_k\})+c(\{v_j,v_k\})$. Therefore $H$ satisfies the simplex condition, and it clearly can be constructed in polynomial time. Call a 2-diverse partition of $\calT$ \emph{good} if every \group in it consists of 2 or 3 vectors with distinct SA values. Lemma~\ref{lem:2diversity} shows that there is an optimal 2-diverse partition that is good. By the construction of $H$, each good 2-diverse partition of $\calT$ can be easily transformed to a perfect matching of $H$ with the same cost, and vice versa. Hence, we can find an optimal 2-diverse partition of $\calT$ by using the polynomial time algorithm for \textsc{Simplex Matching} \cite{stoc07_matching}. We thus have:
\begin{theorem}\label{thm:2diversity}
2-\textsc{Diversity} is solvable in polynomial time.
\end{theorem}

\section{Conclusions}\label{sec:conclu}
This paper presents the first theoretical study on the $t$-closeness principle for privacy preserving.
We prove the NP-hardness of the $t$-\textsc{Closeness} problem for every constant $t\in[0,1)$, and give exact and fixed-parameter algorithms for the problem. We also provide conditionally improved approximation algorithm for $k$-\textsc{Anonymity}, and give the first polynomial time exact algorithm for 2-\textsc{Diversity}.

There are still many related problems that deserve further explorations, amongst which the most interesting one to the authors is designing polynomial time approximation algorithms for $t$-\textsc{Closeness} with provable performance guarantees. We conjecture that the best approximation ratio may be dependent on $n$ (e.g., $O(\log n)$). The parameterized complexity of $t$-\textsc{Closeness} with respect to other sets of parameters are also of interest. Some interesting parameters that have been studied for $k$-anonymity can be found in \cite{joco09,joco_para_kanony,fct11}.

\bibliographystyle{abbrv}
\bibliography{t-closeness}  

\end{document}